\documentclass[11pt,twoside]{article}
\usepackage{fancyhead}
\usepackage{psfrag}
\usepackage{epsfig}
\usepackage{graphicx}
\usepackage{amsmath}
\usepackage{amssymb}
\usepackage{amsthm}
\usepackage{subfigure}
\usepackage{cite}
\usepackage[latin1]{inputenc}

\newtheorem{definition}{Definition}
\newtheorem{construction}{Construction}

\newtheorem{theorem}{Theorem}

\newcommand{\C}{\mathcal{C}}
\newcommand{\I}{\mathcal{I}}
\newcommand{\E}{\mathcal{E}}

\begin{document}

\vspace*{5mm}

\noindent
\textbf{\LARGE Code Constructions based on \\Reed-Solomon Codes
}
\thispagestyle{fancyplain} \setlength\partopsep {0pt} \flushbottom
\date{}

\vspace*{5mm}
\noindent
\textsc{Michael Schelling} \hfill \texttt{michael.schelling@uni-ulm.de} \\
{\small Institute of Communications Engineering, University of Ulm, Germany} \\
\textsc{Martin Bossert} \hfill \texttt{martin.bossert@uni-ulm.de} \\
{\small Institute of Communications Engineering, University of Ulm, Germany} \\

\medskip

\begin{center}
\parbox{11,8cm}{\footnotesize
\textbf{Abstract.} Reed--Solomon codes are a well--studied code class which fulfill the Singleton bound with equality. 
However, their length is limited to the size $q$ of the underlying field $\mathbb{F}_q$. 
In this paper we present a code construction which
yields codes with lengths of factors of the field size. 
Furthermore a decoding algorithm beyond half the minimum distance is given and analyzed.

}
\end{center}

\baselineskip=0.9\normalbaselineskip

\section{Introduction}

\lhead{} \rhead{}
\chead[\fancyplain{}{\small\sl\leftmark}]{\fancyplain{}{\small\sl
\leftmark}} \markboth{\hspace{-2,15cm}Eighth International Workshop
on Optimal Codes and Related Topics\\ July 10-14, 2017, Sofia, Bulgaria \hfill pp. xxx-xxx}{}

	Reed--Solomon (RS) codes were introduced in \cite{RS60} and are well--studied and widely used in various applications. This is due to the fact that RS Codes are maximum distance separable (MDS) codes and due to the existence of efficient decoding algorithms. The classical decoding algorithms are the Peterson algorithm \cite{peterson}, the Berlekamp--Massey algorithm \cite{BM}, and the Sugiyama \textit{et al.} algorithm \cite{Sugiyama}.
	Recently, in 2008, Wu \cite{Wu08} described list decoding algorithms based on an extension of the Berlekamp-Massey algorithm. \\
	A main drawback of these codes is that their length is limited be the size $q$ of the used field $\mathbb{F}_q$. \\
	In \cite{Wu-GII} Wu introduced generalized integrated (GII) RS-Codes which allows to construct longer codes based on RS codes.\\
	The code construction presented in this paper is based on a generalized version of the Plotkin-construction \cite{plotkin} and yields the same length and minimum-distance as Wu's GII codes for an interleaving degree $\nu = 3$.

\newpage

\pagestyle{myheadings} \markboth{OC2017}{Michael Schelling, Martin Bossert}

\section{Code Construction}
	Let $n,q\in\mathbb{N}$ with $n\leq q$ and $q = p^m$.\\
	Consider the RS Codes $ \C_a(n,k_a,d_a), \C_b(n, k_b, d_b)$ and $ \C_z(n, k_z, d_z) $ over the field $\mathbb{F}_q$, such that
	$$\C_z \subseteq \C_b \subseteq \C_a.$$
	This especially implies
	\begin{align*}
		d_a \leq d_b \leq d_z	\text{ and }
		k_a \geq k_b \geq k_z.
	\end{align*}
	 \begin{construction}
	 	The code $\C$ is defined using a generalized \emph{Plotkin-construction}. 
	Let $\alpha \in \mathbb{F}_q\backslash\{0,1\}$.
	\begin{equation}\label{construction}
		\C := \Bigl\{ c=\bigl(c_a|c_b|c_z\bigr) = \bigl(a|a+b|a+\alpha b+z\bigr)\ \Big{|}\ a\in\C_a, b\in\C_b, z\in\C_z\Bigr\}.
	\end{equation}
	 \end{construction}
 	\begin{theorem}
	 	The parameters of the code $\C$ are
	 	\begin{align*}
	 	n_0 &=3n \\
	 	k_0 &= k_a + k_b + k_z \\
	 	d_0 &= \min\{3d_a, 2d_b, d_z\} .
	 	\end{align*}
 	\end{theorem}
 	\begin{proof}
 		The length and dimension are obvious.
 		The minimum distance is proven later by giving a decoding algorithm up to $\left\lfloor\frac{d-1}{2}\right\rfloor$.
	\end{proof}
	
	Comparing the dimension of the code in \eqref{construction} to the one of a GII-RS code of interleaving degree $\nu = 3$ it follows, that the constructed codes are not equivalent.\\
	In the following we choose the parameter $\alpha$ to be a multiplicative generator of the group $\mathbb{F}^*_q$.

\section{Decoding Algorithm}
	\begin{definition}
		Let $c \in\C$  be a transmitted codeword and  $r = c + e$ be received.
		Then $r$ and $e$ consist of the parts
		$$ r = \bigl(r_a | r_b | r_z\bigr), \quad e=\bigl(e_a|e_b|e_z\bigr),$$
		where $r_a, r_b, r_z, e_a, e_b, e_z$ have length $n$ respectively.\\
		
	\end{definition}
	According to the construction of $\C$ the sub code $\C_z$ is the strongest and $\C_a$ is the weakest sub code.
	We take advantage of this property by first decoding in $\C_z$ and then trying to decode in the weaker codes using information from the previous decoding results.
	
	\subsection{Algorithm}
	Let $c = \bigl(a|a+b|a+\alpha b+z\bigr)$ be transmitted and $r = c+e =\bigl(r_a | r_b | r_z\bigr)$ be received.
	Let the the error  $e$ have weight $\tau \leq \lfloor\frac{d_0-1}{2}\rfloor$.
	
	\begin{enumerate}
		\item
		Decode $r_z - \alpha r_b +(\alpha-1) r_a = z + e_z -\alpha e_b + (\alpha-1)e_a$ in $\C_z$.\\
		Let the set of resulting error locations be $\E_{abz}$.
		\item
		Calculate $r_b - r_a = b + e_b - e_a$.\\ 
		Erase all positions from $\E_{abz}$ and decode in a accordingly shortened $\C_b^*$.\\
		Find the corresponding codeword in $\C_b$.
		\item
		Calculate $r_b -b = a + e_b$ and $r_z - \alpha b - z = a + e_z$. \\
		Decode $a + e_a, a + e_b, a + e_z$ in $\C_a$ and receive up to three different solutions $a^{(i)}, i=1,2,3$.
		\item
		Calculate the errors $e^{(i)}_a= r_a - a^{(i)}, e^{(i)}_b = r_b - b- a^{(i)},e^{(i)}_z =  r_z - \alpha b - z - a^{(i)}$ for all $i$ and choose the  $a^{(i)}$ with
		$$\tau_{min} = \min_i \Bigl\{|e^{(i)}_a| + |e^{(i)}_b| + |e^{(i)}_z|\Bigr\}$$ 
		as decoding result.
	\end{enumerate}
\subsection{Correctness of the Algorithm}
\begin{proof}
	The number of errors $\tau$ fulfills
	\begin{align}
	\tau &\leq \left\lfloor\frac{d_z-1}{2}\right\rfloor  \label{eq:d_z}\\
	\tau &\leq \left\lfloor\frac{2d_b-1}{2}\right\rfloor = d_b -1  \label{eq:d_b}\\
	\tau &\leq \left\lfloor\frac{3d_a-1}{2}\right\rfloor. \label{eq:d_a}
	\end{align}
	\begin{enumerate}
		\item	
		The decoding in $\C_z$ is successful as the number of errors in $\E_{abz}^*$ is bounded from above by $\tau$. As overlapping error positions are possible the number of errors is possibly less then $\tau$.
		\item
		The shortened code $\C_b^*$ has parameters $(n_b^* = n-|\E_{abz}^*|, k_b, d_b^* = d_b-|\E_{abz}^*|)$ and has to decode all errors that canceled out in the first step. Let $\I_{abz}$ be the set of these errors.
		For each such erasures of errors at least two errors are necessary, and thus
		\begin{align*}
		|\E_{abz}^*|+2\cdot|\I_{abz}| &\leq \tau.
		\end{align*}
		Together with \eqref{eq:d_b} the decoding of $b$ is successful due to
		\begin{align*}
		|\I_{abz}|&\leq\frac{\tau -|\E_{abz}^*|}{2} \leq \frac{d_b - 1 -|\E_{abz}^*|}{2} =\frac{d_b^*-1}{2}.
		\end{align*}
		\item
		According to \eqref{eq:d_a} it holds that
		$$ \min\Bigl\{|e_a|, |e_b|, |e_z|\Bigr\} < \left\lfloor\frac{d_a-1}{2}\right\rfloor.$$
		W.l.o.g. assume $|e_a|<\left\lfloor\frac{d_a-1}{2}\right\rfloor$.
		Then the decoding of $a + e_a$ in $C_a$ is successful and $a^{(1)}=a$.
		Assume there is an $i$ with $a^{(i)} \neq a^{(1)}$ and
		$$ |e^{(i)}_a| + |e^{(i)}_b| + |e^{(i)}_z| \leq |e^{(1)}_a| + |e^{(1)}_b| + |e^{(1)}_z|. $$
		Then the codewords 
		\begin{align*}
		c^{(1)} = \bigl(a^{(1)}|a^{(1)}| a^{(1)}\bigr)\ \text{ and }\
		c^{(i)} = \bigl(a^{(i)}|a^{(i)}| a^{(i)}\bigr)
		\end{align*}
		are two codewords of a concatenated repetition code with distance
		\begin{align*}
		d(c^{(1)},c^{(i)}) &\leq |e^{(i)}_a| + |e^{(i)}_b| + |e^{(i)}_z| + |e^{(1)}_a| + |e^{(1)}_b| + |e^{(1)}_z| \\
		&\leq 2\tau \stackrel{\eqref{eq:d_a}}{<} 2 \cdot  \left\lfloor\frac{3d_a-1}{2}\right\rfloor < 3d_a .
		\end{align*}
		This is a contradiction to the minimum distance $d_min = 3d_a$ of the concatenated code.
		Thus the decoding of $a$ is successful due to
		$$ 1 = \underset{i}{\mathrm{argmin}} \Bigl\{|e^{(i)}_a| + |e^{(i)}_b| + |e^{(i)}_z|\Bigr\}.$$
	\end{enumerate}
\end{proof}
	
	Note that the number of errors considered in the first step is possibly smaller than the total number of errors $\tau$ due to overlapping error positions. As a result the stated algorithm can decode certain error patterns bevond half the minimum distance.
\section{Simulation}
	The simulation was done for the parameters in table 1 and compared to a MDS code with parameters $(384,216,169)$.	
	\begin{table}[h]\label{tab1}\begin{center}
		\caption{Code Parameters}
		\begin{tabular}{|p{1cm}||p{1cm}|p{1cm}|p{1cm}|}
		\hline
		 \ & n & k & d \\
		 \hline \hline
		 $\C_a$ & 128 & 98 & 31\\
		 \hline
		 $\C_b$ & 128 & 82 & 47\\
		 \hline
		 $\C_z$ & 128 & 36 & 93\\
		 \hline
		 \hline
		 $\C$ &384 & 216 & 93\\
		 \hline
	\end{tabular}
\end{center}
\end{table}

	\includegraphics[width = 1.1\textwidth]{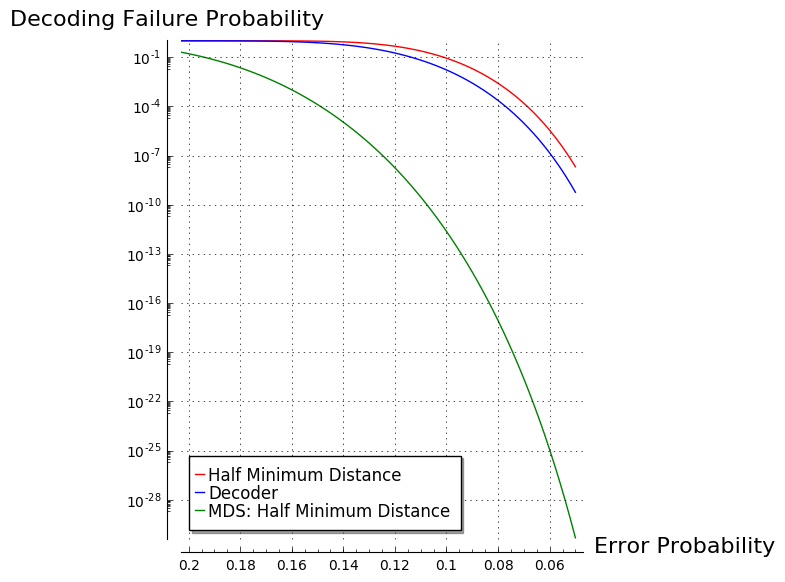}
\section{Conclusion}
	We proposed a new construction for codes based on RS codes, which present a way to construct longer codes over smaller field sizes compared to RS codes. 
	We limited ourselves to the case of three sub codes, a generalization to longer constructions seems possible.\\
	The simulation confirms the capability of the presented decoder to decode beyond half the minimum distance .
	Furthermore the given decoding principles are  straight-forward, can easily be implemented and yield a runtime in the scale of the underlying RS decoder.

\end{document}